\documentclass[conference]{IEEEtran}
\IEEEoverridecommandlockouts

\usepackage{xcolor}
\usepackage[bookmarks=false, colorlinks=true, linkcolor=blue, urlcolor=teal]{hyperref}
\usepackage[utf8]{inputenc} 
\usepackage[T1]{fontenc}
\usepackage{ifthen}
\usepackage{cite,dsfont,proba}
\usepackage[cmex10]{amsmath}

\usepackage{amsmath,amssymb,amsthm}
\usepackage{tikz,enumerate}
\usepackage{cite,mathtools}
\usepackage{hyperref}
\usepackage{cleveref}
\usepackage{mathrsfs}
\usepackage{geometry}
\renewcommand{\baselinestretch}{0.98} 

\theoremstyle{plain}

\newtheorem{theorem}{Theorem}
\newtheorem{lemma}{Lemma}

\newtheorem{corollary}{Corollary}
\newtheorem{definition}{Definition}
\newtheorem{example}{Example}

\newtheorem{remark}{Remark}

\theoremstyle{remark}

\theoremstyle{definition}

\usepackage{tikz}

\title{On the Subadditivity of Mutual f-Information\thanks{This work was supported by the Hong Kong Research Grants Council (RGC) under grant number 14310025, and by the 1+1+1 CUHK–CUHK(SZ)–GDST Joint Collaboration Fund.}} 

\author{
\IEEEauthorblockN{Hamidreza Abin,  Amin Gohari}
\IEEEauthorblockA{Dept. of Information Engineering, CUHK\\
                  \textsuperscript{*}Email: \{hamabin, agohari\}@ie.cuhk.edu.hk}
}
\begin{document}

\maketitle


\begin{abstract}
We address an open problem raised by Polyanskiy and Wu in their textbook concerning the subadditivity of mutual \(f\)-information. We study two notions of generalized mutual information: mutual \(f\)-information, defined by \(\widetilde{I}_f(X;Y)=D_f(p_{XY}\|p_Xp_Y)\), and variational \((G,f)\)-information, obtained by minimizing an average of transformed \(f\)-divergences. For the first notion, we derive a sufficient condition on \(g(x)=x^2f''(x)\) that guarantees subadditivity under conditionally independent observations. This condition yields explicit families beyond the previously known examples generated by \(x\log x\) and \((x-1)\log x\). Our proof relies on the concavity of a generalized variance functional. For the variational notion, we extend existing sufficient conditions beyond the transforms \(G(x)\in\{x,\log(1+x),-\log(1-x)\}\). Using a comparison inequality for \(G^{-1}(G(A)+G(B))\), we obtain differential criteria for concave and convex transforms. Explicit constructions illustrate the resulting classes of subadditive information measures.
\end{abstract}

\section{Introduction}

Mutual information is a fundamental tool in information theory, statistics, and machine learning for quantifying dependence between random variables. Its usefulness stems not only from its operational interpretations but also from structural properties such as data processing and subadditivity. A natural approach to generalizing mutual information is to replace the Kullback--Leibler divergence with an $f$-divergence. Introduced by Ali--Silvey and Csisz\'ar~\cite{ali1966general,csiszar1967information}, the $f$-divergence between distributions $p\ll q$ on a finite alphabet is\footnote{We say that \(p_X\) is absolutely continuous with respect to \(q_X\), denoted by \(p_X \ll q_X\), if for every \(x \in \mathcal{X}\),
$q_X(x) = 0 \implies p_X(x) = 0$.}
\[
D_f(p\|q)=\sum_x q(x)f\!\left(\frac{p(x)}{q(x)}\right),
\]
where $f$ is convex and satisfies $f(1)=0$. This framework includes many familiar discrepancy measures and provides several distinct generalizations of Shannon mutual information. In particular, \(  \widetilde{I}_{f}(X;Y)\) reduces to Shannon's mutual information when we choose \(f(x)=x\log x\).

In this paper, we investigate subadditivity for two such generalizations. The first is the mutual $f$-information
\[
\widetilde{I}_f(X;Y)
\triangleq D_f(p_{XY}\|p_Xp_Y).
\]
The second is the variational $(G,f)$-information
\[
I_{G,f}(X;Y)
\triangleq
\min_{q_Y}\sum_x p_X(x)
G\!\left(D_f(p_{Y|X=x}\|q_Y)\right),
\]
where $G$ is non-decreasing on the relevant divergence range and $G(0)=0$. Both quantities reduce to Shannon mutual information for $f(x)=x\log x$ and, in the second definition, $G(x)=x$. In general, however, they are different: $\widetilde{I}_f$ evaluates the divergence relative to the actual product of marginals, whereas $I_{G,f}$ optimizes an auxiliary output distribution. This distinction is important when studying subadditivity.

\subsection*{Subadditivity of mutual $f$-information}

For Shannon mutual information, conditionally independent observations satisfy
\[
I(X;YZ)\leq I(X;Y)+I(X;Z)
\qquad\text{whenever }Y\to X\to Z.
\]
This inequality allows information obtained from multiple observations to be bounded by a sum of individual contributions. 
We call this property ``subadditivity'' because it is equivalent to subadditivity of mutual information over product channels:
\[
I(X_1X_2;Y_1Y_2)\leq I(X_1;Y_1)+I(X_2;Y_2)
\]
for every \(p(x_1,x_2)\,p(y_1|x_1)\,p(y_2|x_2)\).

Its counterpart for mutual $f$-information is
\begin{equation}
\widetilde{I}_f(X;YZ)
\leq
\widetilde{I}_f(X;Y)+\widetilde{I}_f(X;Z),
\label{intro:mutual-f-subadditivity}
\end{equation}
whenever $Y\to X\to Z$. In particular, this property yields subadditivity over product channels, even when their inputs are correlated.

Identifying the functions $f$ for which
\eqref{intro:mutual-f-subadditivity} holds is a nontrivial problem. 
In their textbook~\cite[Section~7.8]{polyanskiy2025information}, Polyanskiy and Wu mention that the functions $f(x)=x\log x$ and $f(x)=(x-1)\log x$ satisfy this property  and pose the existence of further examples as an open question. 

Our first contribution is a sufficient condition that produces additional families of subadditive mutual $f$-information. Writing
\[
g(x)=x^2f''(x),
\]
we show that subadditivity follows if $g''(x)<0$ for all $x>0$ and there exist constants $\eta_g>0$ and $\ell_g\in\mathbb{R}$ such that
\[
g(x)+\frac{2(g'(x)-\eta_g)^2}{-g''(x)}
\leq
2\eta_g+(x-1)\ell_g,
\qquad x>0.
\]
The proof reduces the desired information inequality to concavity, with respect to $q$, of the generalized variance functional
\[
q\longmapsto
\mathbb{E}_{p}\!\left[
g\!\left(\frac{q}{p}\right)
\left(u-\mathbb{E}_{q}[u]\right)^2
\right],
\]
where $p$ is an arbitrary distribution and $u$ is an arbitrary function.
Although the displayed differential condition excludes affine $g$, the underlying concavity argument also recovers the two classical generators.

We further give a construction based on functions of the form $g(x)=\alpha x+s(x)$. This leads to explicit subadditive generators, including
\[
f(x)=x\log x+\frac{2}{3}(\sqrt{x}-1)^2
\]
and
\[
f(x)=2x\log x-(x+1)\log\!\left(\frac{x+1}{2}\right).
\]
Additional examples follow from integration over parametric families and from other choices of $s$. These results provide sufficient conditions and new examples, rather than a complete characterization of all subadditive generators.

\subsection*{Subadditivity of variational $(G,f)$-information}

Our second contribution concerns the variational framework developed in~\cite{abin2026class:isit,abin2026class:arx}. Define the transformed divergence
\[
\mathcal{D}_{G,f}(p\|q)
\triangleq G\!\left(D_f(p\|q)\right).
\]
The cited work establishes that divergence subadditivity over product distributions implies subadditivity of $I_{G,f}$, with the converse holding under mild regularity and non-degeneracy assumptions. It also develops binary-alphabet reductions and sufficient conditions for the transforms
\[
G(x)\in\{x,\log(1+x),-\log(1-x)\}.
\]

Here, we extend the available sufficient conditions for variational information to a broader class of transforms. We consider
\[
G(x)=\int_0^x\frac{1}{u(t)}\,dt,
\]
with positive concave $u$ on an appropriate domain. The central step is a comparison bound for
$G^{-1}(G(A)+G(B))$ in terms of $A+B$ and $AB$. Combining this bound with scaled versions of the logarithmic criteria yields sufficient conditions expressed through
\[
x^2g''(x)-\gamma g(\alpha)g(x/\alpha)\leq 0
\]
or
\[
x^2g''(x)+\gamma g(\alpha)g(x/\alpha)\leq 0,
\]
where $g(x)=x^2f''(x)$ and $\gamma>0$ is chosen according to the transform and the range of $D_f$. 

\subsection*{Related work and organization}

The two notions of information studied here have different subadditivity requirements. Csisz\'ar's variational formulation~\cite{csiszar1972class} and its transformed extension~\cite{abin2026class:isit,abin2026class:arx} motivate our second set of results, while the question in~\cite{polyanskiy2025information} motivates the first. Other tensorization results for mutual $f$-information involve different assumptions and inequality directions. For example, concavity of $1/f''$ yields superadditivity properties under independent inputs~\cite{masiha2023f} and tensorization of the $f$-hypercontractivity ribbon~\cite{beigi2018phi,wang2022estimation22}. Our focus is instead on subadditivity under conditionally independent observations. For additional properties and applications of \(f\)-divergences, we refer the reader to \cite{aa1,aa2,aa3,aa4,aa5,aa6,aa7,aa8,aa9,qq1,qq2,qq3,qq4,qq5,qq6,qq7,qq8,qq9,qq10}.

Section~II develops the sufficient condition for $\widetilde{I}_f$ and presents explicit examples. Section~III reviews the variational $(G,f)$ framework and establishes the extensions to other transforms. Throughout, we restrict attention to finite alphabets. Detailed proofs and calculations are provided in the appendices.

\subsection*{Notation}
Sets and alphabets are denoted by calligraphic letters (e.g., \(\mathcal{X}\)). Random variables are represented by uppercase letters (e.g., \(X, Y\)), while their specific realizations are written in the corresponding lowercase letters (e.g., \(x, y\)). Functions are denoted by either lowercase or uppercase letters (e.g., \(f(x)\) or \(G(x)\)). 

All probability distributions are assumed to be defined over finite alphabets. The marginal distribution of a random variable \(X\) is denoted by a letter with subscript \(X\) (such as \(p_X\) or \(Q_X\)), and joint distributions over product spaces like \(\mathcal{X} \times \mathcal{Y}\) are indexed accordingly (e.g., \(p_{XY}\)). Conditional distributions are indicated with a conditioning subscript, such as \(p_{Y|X}\) or \(q_{Y|X}\). 

A Markov chain relationship is denoted by \(X \to Y \to Z\), signifying that \(X\) and \(Z\) are conditionally independent given \(Y\), which implies the factorization \(P_{XZ|Y} = P_{X|Y}P_{Z|Y}\).
Finally, throughout this paper, the symbol \(\log\) is used to denote the natural logarithm.

\color{black}


\section{ Mutual \(f\)-Information}

\begin{definition}
    Consider a convex function \(f:[0,\infty)\rightarrow \mathbb{R}\cup \{+\infty\}\) with \(f(1)=0\). We say that the function \(f\) is subadditive if, for all random variables \(X\), \(Y\), and \(Z\) defined over finite sets, we have
    \begin{align}
 \widetilde{I}_{f}(X;YZ)&\leq \widetilde{I}_{f}(X;Y) + \widetilde{I}_{f}(X;Z)\label{sub_I_def}
\end{align}
whenever \(Y \rightarrow X \rightarrow Z\) forms a Markov chain, i.e., \(p_{XYZ} = p_X p_{Y|X} p_{Z|X}\).
\end{definition}
For applications of subadditive $f$-information functions, see \cite{abin2026class:isit}. In the next theorem, we introduce a class of subadditive functions.
\begin{theorem}\label{thm:mutual-f-subadditivity}
    Let \(f:[0,\infty)\rightarrow \mathbb{R}\) with \(f(1)=0\) be a continuous convex function on $[0,\infty)$. Assume that $f$ is also four times continuously differentiable on $(0,\infty)$. Let \(g(x)=x^2f''(x)\). If \(g''(x)<0\) for all \(x>0\) and there exist constants \(\eta_g>0\) and \(\ell_g\in\mathbb{R}\) such that
    \begin{align}
       g(x)+\dfrac{2(g'(x)-\eta_g)^2}{-g''(x)}\leq 2\eta_g+(x-1)\ell_g~~\forall x>0,
    \end{align}
    then \(\widetilde{I}_{f}\) is subadditive.
\end{theorem}

\begin{proof}
      The proof can be found in Appendix~\ref{app:mutual-f-subadditivity}.
\end{proof}
\begin{corollary}\label{coro:1}
Let \(s(x)\) be a positive concave function on \(x>0\) such that \(s''(x)<0\) for all \(x>0\). Assume that the function
\[
\mathtt{K}_{s}(x)= s(x)+\frac{2\left(s'(x)\right)^2}{-s''(x)}
\]
is concave. Let \(\alpha\) be a positive number such that
\[
\alpha \geq \mathtt{K}_{s}(1).
\]
Let \(g(x)=\alpha x+s(x)\). Then, by choosing
\[
(\eta_g,\ell_g)=\left(\alpha,\; \alpha+ \frac{d}{dx}\mathtt{K}_{s}(x)\Big|_{x=1}\right),
\]
the function \(g(x)\) satisfies the assumptions of Theorem~\ref{thm:mutual-f-subadditivity}.

Moreover, the corresponding \(f(x)\) is given by
\begin{align}
        f(x)=\alpha x\log x+\int_{1}^{x}(x-t)\frac{s(t)}{t^2}\,dt,\label{int:1:1}
\end{align}
for which  \(\widetilde{I}_{f}\) is subadditive.
\end{corollary}
\begin{proof}
  See Appendix~\ref{coro:1:proof}.
\end{proof}
If the functions $f_1(x)$ and $f_2(x)$ lead to subadditive information measures $\widetilde{I}_{f_1}$ and $\widetilde{I}_{f_2}$, then for any non-negative constants $\alpha, \beta \ge 0$, the function $f(x) = \alpha f_1(x) + \beta f_2(x)$ also leads to a subadditive information measure. Thus, we can mix the functions obtained from Corollary~\ref{coro:1} (as well as $x\log x$ and $(x-1)\log(x)$).

\begin{example}\label{ex;1}
    The following functions are subadditive:
  \begin{enumerate}
      \item $f(x)=x\log(x)+\frac{2\alpha}{3}(\sqrt{x}-1)^{2}$ for $\alpha\in[0,1]$;
      \item $f(x)=\frac{5}{6}x\log(x)+1-\frac{x-1}{\log(x)}$;
      \item $f(x)=2x\log(x)-(x+1)\log\left(\frac{x+1}{2}\right)$;
      \item $f(x)=x\log(x)+\frac{1}{4}\int_{1}^{x}(x-t)\frac{\log(1+t)}{t^2}dt$.
  \end{enumerate}
\end{example}

For details of the calculations in Example~\ref{ex;1}, please see Appendix~\ref{ex;1:p}.

\section{Variational Mutual \(f\)-Information}

In this section, we consider a different variational form of mutual \(f\)-information.
In this part, we aim to extend the results of \cite{abin2026class:isit}, whose full arXiv version is available in \cite{abin2026class:arx}.
\begin{lemma}[\cite{abin2026class:isit}]\label{domain:lem}
For arbitrary probability distributions \(p_X\) and \(q_X\) defined on a set \(\mathcal{X}\), the following inequality holds:
\begin{align}
    D_f(p_X\|q_X) \leq \mathsf{D_m}(f) \triangleq f(0) + \lim_{t\rightarrow\infty}\frac{f(t)}{t}. \label{defDmf}
\end{align}
\end{lemma}

\begin{definition}[\cite{abin2026class:isit}]\label{def:admissible}
A pair of functions \((G(x),f(x))\) is said to be admissible if \(f:[0,\infty)\rightarrow\mathbb{R}\) is convex and satisfies \(f(1)=0\), while \(G:[0,\mathsf{D_m}(f))\rightarrow [0,\infty)\) is non-decreasing and satisfies \(G(0)=0\). Here, \(\mathsf{D_m}(f)\) is the quantity introduced in \eqref{defDmf}.
\end{definition}

\begin{definition}[\cite{abin2026class:isit}]\label{def:div:mutual}
Let \((G,f)\) be an admissible pair. For distributions \(p_X\ll q_X\) on a finite alphabet \(\mathcal{X}\), the \((G,f)\)-divergence \(\mathcal{D}_{G,f}\) is defined as
\begin{align}
\mathcal{D}_{G,f}(p_X\|q_X)
&\triangleq G\!\big(D_f(p_X\|q_X)\big).
\end{align}
\end{definition}

\begin{definition}[\cite{abin2026class:isit}]
For random variables \(X\) and \(Y\), the corresponding mutual information \(I_{G,f}(X;Y)\) is defined by
\begin{align*}
I_{G,f}(X;Y)
&\triangleq \min_{q_Y} \sum_x p_X(x)\,\mathcal{D}_{G,f}(p_{Y|X=x}\|q_Y)
\\&= \min_{q_Y} \sum_x p_X(x)\, G\!\big(D_f(p_{Y|X=x}\|q_Y)\big).
\end{align*}
\end{definition}
\begin{remark}
It follows immediately that \(I_{G,f}(X;Y)\) coincides with Shannon's mutual information when \(G(x)=x\) and \(f(x)=x\log x\). Furthermore, choosing
\[
G(x)=\frac{1}{\alpha-1}\log\bigl(1+(\alpha-1)x\bigr),
\qquad
f(x)=\frac{x^\alpha-1}{\alpha-1}
\]
recovers Rényi information for \(\alpha\in(0,1)\cup(1,\infty)\).
In general, however, \(I_{G,f}(X;Y)\) need not be symmetric, i.e.,
\[
I_{G,f}(X;Y) \neq I_{G,f}(Y;X),
\]
which differs from the behavior of Shannon's mutual information. We refer to \cite{abin2026class:arx} for several basic properties of \(I_{G,f}\).
\end{remark}
\begin{definition}[\cite{abin2026class:isit}]
An admissible pair \((G,f)\) is said to be \emph{divergence-subadditive} if, for any distributions \(q_Y\ll r_Y\) on a finite set \(\mathcal{Y}\) and any distributions \(q_Z\ll r_Z\) on a finite set \(\mathcal{Z}\), the following inequality holds:
\begin{align}
  \mathcal{D}_{G,f}(q_Y q_Z \,\|\, r_Y r_Z)
  \;&\leq\;
  \mathcal{D}_{G,f}(q_Y \,\|\, r_Y)\nonumber\\
  &+
  \mathcal{D}_{G,f}(q_Z \,\|\, r_Z). \label{gen-def-1}
\end{align}
Moreover, a pair \((G,f)\) is called \emph{information-subadditive} if, for any random variables \(X,Y,Z\) defined on finite sets, we have
\begin{align}
  I_{G,f}(X;YZ)&\leq I_{G,f}(X;Y) + I_{G,f}(X;Z)\label{sub_I_def:2}
\end{align}
whenever \(Y \rightarrow X \rightarrow Z\) forms a Markov chain, i.e., \(p_{XYZ} = p_X p_{Y|X} p_{Z|X}\).
\end{definition}
The following theorem establishes that, under certain mild assumptions, information-subadditivity and divergence-subadditivity are equivalent.
\begin{theorem}[\cite{abin2026class:isit}]\label{gen-thm1}
    Let \((G,f)\) be an admissible pair (see Definition \ref{def:admissible}), where \(G\) is a continuous function satisfying \(G(x)>0\) for all \(x>0\), and \(f(x)\) is strictly convex at \(x=1\). Then \((G,f)\) is \emph{divergence-subadditive} if and only if it is \emph{information-subadditive}.
\end{theorem}
Although inequalities \eqref{gen-def-1} and \eqref{sub_I_def:2} are equivalent, it is more convenient to work with \eqref{gen-def-1}, since the definition of \(I_{G,f}\) involves a minimization problem. The following result from \cite{abin2026class:isit} shows that, for a certain class of functions \(G\), the subadditivity of \(\mathcal{D}_{G,f}\) over arbitrary finite alphabets already follows from its subadditivity for binary random variables.
\begin{lemma}[\cite{abin2026class:isit}]\label{gen-lemma1}
Let $G:[0,\nu)\rightarrow[0,+\infty)$ be an increasing function with $G(0)=0$ such that the range of $G$ is $[0, \infty)$. Moreover, assume that for every $c>0$, the function
\begin{align}
x \mapsto G^{-1}\big(G(x) + G(c)\big)\label{cond:G}
\end{align}
is concave on $[0,\nu)$.  
Let $f$ be a convex function on $[0,\infty)$ with $f(1)=0$ such that $\mathsf{D_m}(f)\leq \nu$, where $\mathsf{D_m}(f)$ is defined in \eqref{defDmf}.  Assume that \eqref{gen-def-1} holds for any distributions $q_Y,r_Y$ on $\mathcal{Y}$ and $q_Z,r_Z$ on $\mathcal{Z}$ when both alphabets $\mathcal{Y}$ and $\mathcal{Z}$ are binary (i.e., of size two).  
Then \eqref{gen-def-1} also holds for this pair and for any distributions $q_Y,r_Y$ and $q_Z,r_Z$ defined on finite alphabets $\mathcal{Y}$ and $\mathcal{Z}$ of arbitrary size.
\end{lemma}
\begin{lemma}[\cite{abin2026class:arx}]\label{corol:g}
Let \(u(t)\) be any positive concave function on \([0,\nu)\). Define a function \(G(x)\) on \([0,\nu)\) by
\[
G(x) = \int_0^x \frac{1}{u(t)}\,dt.
\]
Then \(G\) satisfies the assumptions in Lemma~\ref{gen-lemma1}.
\end{lemma}

\subsection{New Results on the Subadditivity of \(I_{G,f}\)}
For a given function \(G\), our goal is to characterize the functions \(f\) such that the admissible pair \((G,f)\) is divergence-subadditive (equivalently, information-subadditive). In \cite{abin2026class:isit}, the authors characterized subadditive pairs \((G,f)\) when
\(G(x)\in \{x,\log(1+x),-\log(1-x)\}\). In this paper, we modify their proof to characterize subadditive pairs \((G,f)\) for other choices of the function \(G\). The following theorem characterizes divergence-subadditive pairs \((G,f)\) when the function \(G(x)\) is concave.
\begin{theorem}\label{th:concave:G}
Assume that \[
G(x) = \int_0^x \frac{1}{u(t)}\,dt
\]
for a positive increasing concave function \(u(t)\) on \([0,\infty)\).  Let \(f\) be a four-times continuously differentiable convex function with \(f(1)=0\). For a positive parameter \(\gamma\), define the set
\begin{align}
    \mathcal{T}_{\gamma}=\Bigl\{g:[0,\infty)\rightarrow [0,\infty)\;\Big|\;&x^2g''(x)-\gamma\, g(\alpha)\,g\!\left(\frac{x}{\alpha}\right)\leq 0,\nonumber\\& \forall x>0,~\forall\alpha>0\Bigr\}.
\end{align}
If \(x^2f''(x)\in \mathcal{T}_{\gamma}\) for
\[
0<\gamma\leq \frac{G^{-1}\!\bigl(2G(\mathsf{D_m}(f))\bigr)-2\mathsf{D_m}(f)}{\mathsf{D_m}^{2}(f)},
\]
then the pair \((G,f)\) is divergence-subadditive. Here, \(\mathsf{D_m}(f)\) is as defined in \eqref{defDmf}.
\end{theorem}
\begin{proof}
    The proof can be found in Appendix~\ref{th:concave:G:p}.
\end{proof}

\begin{remark} The set \(\mathcal{T}_{\gamma}\) for $\gamma=1$ was already introduced in \cite{abin2026class:isit} when studying the special case of \(G(x)=\log(1+x)\). Several parametric families of functions belonging to \(\mathcal{T}_{\gamma}\) for $\gamma=1$ were also given in \cite{abin2026class:arx}. One may reuse these parametric functions for \(\mathcal{T}_{\gamma}\) when $\gamma\neq 1$ after suitable adjustments. 
\end{remark}
\begin{example}\label{ex:con:hh}For \(u(x)=\frac{x^{1-p}}{p}\) on \([0,+\infty)\), we obtain \(G(x) = x^p\) for \(0 < p \leq 1\). Set \(\widetilde{a}_p=5\bigl(2^{1/p}-2\bigr)\) and
\[
\Theta_p=\frac{-12-\widetilde{a}_p+\sqrt{24\widetilde{a}_p+169}}{5-\widetilde{a}_p}.
\]
Choose \(\epsilon\) such that \(\frac{1}{5}<\epsilon\leq \Theta_p\). Then, for \(\gamma=\frac{\widetilde{a}_p}{20+4\epsilon}\), the non-concave and non-negative function
\[
\ell_{\epsilon}(x)=\sqrt{x}\left(1+\epsilon \cos\left(\log x\right)\right)
\]
belongs to \(\mathcal{T}_{\gamma}\), and the corresponding function
\[
f_{\epsilon}(x)=(2+0.4\epsilon)\left(\sqrt{x}-1\right)^2+0.8\epsilon \sqrt{x}\left(1-\cos\left(\log x\right)\right)
\]
is such that the pair \((x^p,f_{\epsilon})\) is divergence-subadditive.
\end{example}
For details of this example, please see Appendix \ref{ex:con:hh:p}.

The following theorem characterizes divergence-subadditive pairs \((G,f)\) for a class of convex functions \(G(x)\).
\begin{theorem}\label{th:convex:G}
Assume that \[
G(x) = \int_0^x \frac{1}{u(t)}\,dt
\]
for a positive decreasing concave function \(u(t)\) on \([0,\nu)\) for some $\nu>0$. Then \(G(x)\) is convex. Let \(f\) be a four-times continuously differentiable convex function with \(f(1)=0\). For a positive parameter \(\gamma\), define the set
\begin{align}
    \mathcal{T}_{\gamma}^{-}=\Bigl\{g:[0,\infty)\rightarrow [0,\infty)\;\Big|\;&x^2g''(x)+\gamma\, g(\alpha)\,g\!\left(\frac{x}{\alpha}\right)\leq 0,\nonumber\\& \forall x>0,~\forall\alpha>0\Bigr\}.
\end{align}
If \(x^2f''(x)\in \mathcal{T}_{\gamma}^{-}\) for
\[
\gamma\geq  \frac{2\mathsf{D_m}(f)-G^{-1}\!\bigl(2G(\mathsf{D_m}(f))\bigr)}{\mathsf{D_m}^{2}(f)},
\]
then the pair \((G,f)\) is divergence-subadditive. Here, \(\mathsf{D_m}(f)\) is as defined in \eqref{defDmf}.
\end{theorem}
\begin{proof}
    The proof can be found in Appendix \ref{th:convex:G:p}.
\end{proof}
\begin{remark}
Theorem \ref{th:convex:G} suggests that, in order to find a suitable function \(f\), one should characterize the set \(\mathcal{T}_{\gamma}^{-}\). Suppose that \(\ell_{\gamma}(x)\in \mathcal{T}_{\gamma}^{-}\). Then the equation \(x^2 f_\gamma''(x)=\ell_{\gamma}(x)\) yields the following solution:
\begin{align}
    f_{\gamma}(x)=\widetilde{a}(x-1)+\int_{1}^{x}(x-t)\frac{\ell_{\gamma}(t)}{t^2}\,dt.
\end{align}
Note that the parameter \(\mathsf{D_m}(f)\) depends on \(\gamma\). After computing \(f_{\gamma}(x)\) using the formula above, one should calculate \(\mathsf{D_m}(f_{\gamma})\) and then choose the parameter \(\gamma\) such that
\[
\gamma\geq \frac{2\mathsf{D_m}(f_{\gamma})-G^{-1}\!\bigl(2G(\mathsf{D_m}(f_{\gamma}))\bigr)}{\mathsf{D_m}^{2}(f_{\gamma})}.
\]
The most effective approach for designing the function \(f\) is to consider a parametric family in \(\mathcal{T}_{\gamma}^{-}\) and then select the parameters so that all conditions are satisfied.
\end{remark}

\begin{remark} The set \(\mathcal{T}^{-}_{\gamma}\) for $\gamma=1$ was already introduced in \cite{abin2026class:isit} when studying the special case of \(G(x)=-\log(1-x)\).  
\end{remark}

\begin{example}\label{ex:convex:rt}For \(u(x)=1-\frac{x^2}{\nu^2}\) on \([0,\nu)\), we obtain \(G(x)=\frac{\nu}{2}\log\left(\frac{\nu+x}{\nu-x}\right)\).
As an example, set \(G(x)=\frac{1}{2}\log\left(\frac{1+x}{1-x}\right)\), which corresponds to $\nu=1$. Let \(0<\epsilon<\frac{1}{25}\). Define
\[
m_\epsilon=\sqrt{\frac{(1+\epsilon)(1-25\epsilon)}{75\epsilon^2+44\epsilon+1}},
\qquad
\gamma=\frac{2m_{\epsilon}}{m_{\epsilon}^2+1}.
\]
Then the non-negative and concave function
\[
\ell_{\epsilon}(x)=\frac{m_{\epsilon}\sqrt{x}}{4(1+\epsilon)}\Bigl(1+5\epsilon \cos(\log x)\Bigr)
\]
belongs to \(\mathcal{T}_{\gamma}^{-}\), and the corresponding function
\[
f_{\epsilon}(x)=m_{\epsilon}\left(1-\frac{\sqrt{x}\,\bigl(1+\epsilon\cos(\log x)\bigr)}{1+\epsilon}\right)
\]
is such that the pair \(\bigl(\tfrac{1}{2}\log\bigl(\tfrac{1+x}{1-x}\bigr),f_{\epsilon}\bigr)\) is divergence-subadditive.
\end{example}
For details of this example, please see Appendix~\ref{ex:convex:rt:p}.
\begin{remark}
It is well known that if \(G(x)\) is concave on \(x>0\) with \(G(0)=0\), then \(G(2x)\leq 2G(x)\); similarly, if \(G(x)\) is convex on \(x>0\) with \(G(0)=0\), then \(G(2x)\geq 2G(x)\). When \(G(x)\) satisfies \(G(0)=0\) but is neither convex nor concave, both inequalities may hold, though on different regions of \(x\).
\end{remark}
\begin{remark}
In Theorems \ref{th:concave:G} and \ref{th:convex:G}, the concavity and convexity assumptions on the function \(G\) are imposed only to guarantee that the parameter \(\gamma\) is positive; see the proofs of these theorems. Thus, Theorems \ref{th:concave:G} and \ref{th:convex:G} remain valid whenever one can find a positive \(\gamma\) satisfying the assumptions of these theorems. For other cases of the function \(G\), see the following example.
\end{remark}
\begin{example}\label{neither:conv:1}
For \(u(x)=4-(x-1)^2\) on \([0,3)\), we obtain
\[
G(x)=\frac{1}{4}\log\left(\frac{3(1+x)}{3-x}\right).
\]
The function \(G\) is neither convex nor concave.

Consider the setting of Example~\ref{ex:convex:rt}, but assume that
\[
G(x)=\frac{1}{4}\log\left(\frac{3(1+x)}{3-x}\right), \qquad x\in[0,3).
\]
All other parameters and functions (except \(G\)) are the same as in the setting of Example~\ref{ex:convex:rt}. We show in Appendix~\ref{neither:conv:1:p} that the pair
\[
\left(\frac{1}{4}\log\left(\frac{3(1+x)}{3-x}\right),f_{\epsilon}\right)
\]
is divergence-subadditive. Here, \(f_{\epsilon}\) is defined in Example~\ref{ex:convex:rt}.
\end{example}
\textbf{AI statement:} The authors have worked on the subadditivity problem for more than two years and have previously published research on this topic. Guided by their mathematical intuition, they used AI through a closely supervised, iterative process of prompting, evaluation, and refinement to develop the results presented here. All AI-generated content was carefully reviewed and verified by the authors.
\color{black}

\bibliographystyle{ieeetr}
\bibliography{biblioisit}
\appendix
\subsection{Proof of Theorem~\ref{thm:mutual-f-subadditivity}}
\label{app:mutual-f-subadditivity}

For fixed finite alphabets $\mathcal{X},\mathcal{Y}$, \(\widetilde{I}_f(X;Y)\) is a continuous function of $p_{XY}$. Therefore, if \(f\) is subadditive on all full-support distributions, then \(f\) is subadditive over all  distributions. 

We begin by computing \(\widetilde{I}_f(X;YZ)-\widetilde{I}_f(X;Z)\). For given \(p_{Z}\), \(p_{X|Z}\), and \(p_{Y|X}\), let \(p_{X,Y,Z}=p_{Z}\,p_{X|Z}\,p_{Y|X}\). We have
\begin{align}
    \widetilde{I}_f\left(X;YZ\right)-\widetilde{I}_f(X;Z)=\sum_{z\in \mathcal{Z}}p_{Z}(z)\,\Theta_{p_X,p_{Y|X}}\left(p_{X|Z=z}\right),
\end{align}
where, for arbitrary $q_X, q_{Y|X}$, and $r_X$, \(\Theta_{q_X,q_{Y|X}}\left(r_{X}\right)\) is defined as
\begin{align}
&\Theta_{q_X,q_{Y|X}}\left(r_{X}\right)\triangleq\nonumber
\\&\sum_{y\in\mathcal{Y}}r_{Y}(y)\left[D_{f}\left(r_{X|y}\,\Big\|\,q_X\right)\right]-D_{f}\bigl(r_{X}\|q_X\bigr),
\end{align}
and $r_{YX}=r_Xq_{Y|X}$.
Lemma \ref{lemma4} shows that for fixed \(q_X\) and \(q_{Y|X}\), the function \(\Theta_{q_X,q_{Y|X}}\left(r_X\right)\) is a concave function of \(r_{X}\). Thus, Jensen's inequality implies
\begin{align}
    &\sum_{z\in \mathcal{Z}}p_{Z}(z)\,\Theta_{p_X,p_{Y|X}}\left(p_{X|Z=z}\right)
    \\&\leq \Theta_{p_X,p_{Y|X}}\left(\sum_{z\in \mathcal{Z}}p_{Z}(z)p_{X|Z=z}\right)
    \\&=\Theta_{p_X,p_{Y|X}}\left(p_{X}\right)
    \\&=\tilde{I}_f(X;Y).
\end{align}
This completes the proof. 
\begin{lemma}\label{yhn:1}
Under the assumptions of Theorem~\ref{thm:mutual-f-subadditivity},
for any fixed  $p_{X}$ and $p_{Y|X}$, the map \(\Theta_{p_X,p_{Y|X}}\left(r_X\right)\) is a concave function of \(r_{X}\), 
where \(\Theta_{p_X,p_{Y|X}}\left(r_{X}\right)\) is defined as
\begin{align}
&\Theta_{p_X,p_{Y|X}}\left(r_{X}\right)\triangleq\nonumber
\\&\sum_{y\in\mathcal{Y}}r_{Y}(y)\left[D_{f}\left(r_{X|y}\,\Big\|\,p_X\right)\right]-D_{f}\bigl(r_{X}\|p_X\bigr)
\end{align}
and $r_{YX}=r_Xp_{Y|X}$.\label{lemma4}
\end{lemma}
\begin{proof}
Since \(r_X\) is a probability distribution, it suffices to show that for every function \(v_X\) satisfying \(\sum_{x}v_X(x)=0\), the following quantity is non-positive:
\begin{align}
    &\frac{d^2}{d\epsilon^2}\Theta_{p_X,p_{Y|X}}\left(r_X+\epsilon v_X\right)\Big|_{\epsilon=0}
\nonumber\\&=\vec{v}_{X}^{T}\,\mathcal{H}\!\left(\Theta_{p_X,p_{Y|X}}\left(r_X\right)\right)\vec{v}_{X}.
\end{align}
Here, \(\mathcal{H}\!\left(\Theta_{p_X,p_{Y|X}}\left(r_X\right)\right)\) denotes the Hessian matrix of \(\Theta_{p_X,p_{Y|X}}\left(r_X\right)\), and \(\vec{v}_{X}\) is the column vector corresponding to the function \(v_X\), while \(T\) denotes the transpose operator. Let \(g(x)\triangleq x^2f''(x)\)
and set
$$\beta_X(x)=\frac{v_X(x)}{r_X(x)}, \qquad\forall x.$$
Then,  routine calculations show that
\begin{align}
     &\frac{d^2}{d\epsilon^2}\Theta_{p_X,p_{Y|X}}\left(r_X+\epsilon v_X\right)\Big|_{\epsilon=0} \nonumber\\&=
     \sum_{y\in\mathcal{Y}}r(y)\Bigg[\sum_{x}p(x)\,g\!\left(\frac{r(x|y)}{p(x)} \right)\nonumber\\&\quad\times\left(\beta(x)-\sum_{x'\in \mathcal{X}}r(x'|y)\,\beta(x')\right)^2\Bigg]\nonumber\\&-\sum_{x\in\mathcal{X}}p(x)\,g\!\left(\frac{r(x)}{p(x)}\right)\beta^{2}(x).
\end{align}
For every real-valued function \(u_X\) and all distributions \(p_X\) and $q_X$, define the function
\begin{align}
    \mathcal{L}_{u_X,p_X}(q_X)\triangleq \mathbb{E}_{p_X}\!\left[g\!\left(\frac{q_X}{p_X}\right)\left(u_X-\mathbb{E}_{q_X}(u_X)\right)^2\right].
\end{align}
Then, one can observe that
\begin{align}
     &\frac{d^2}{d\epsilon^2}\Theta_{p_X,p_{Y|X}}\left(r_X+\epsilon v_X\right)\Big|_{\epsilon=0} \nonumber\\&=
 \sum_{y\in\mathcal{Y}}r(y)\mathcal{L}_{\beta_X,p_X}(r_{X|y}) \\&-\mathcal{L}_{\beta_X,p_X}(r_{X}).
\end{align}
Thus, it suffices to show that $\mathcal{L}_{u_X,p_X}(q_X)$ is concave in  $q_X$ for any $u_X$ and $p_X$. 

 As before, consider an arbitrary real-valued function \(k_X\) such that \(\sum_{x\in \mathcal{X}}k_X(x)=0\). To show that \(\mathcal{L}_{u_X,p_X}(q_X)\) is concave over probability distributions, it suffices to show that
\begin{align}
        &\frac{d^2}{d\epsilon^2}\mathcal{L}_{u_X,p_X}(q_X+\epsilon k_X)\Big|_{\epsilon=0} \leq 0
\end{align}
for all distributions \(p_X\) and \(q_X\) and all functions \(u_X\) and \(k_X\) satisfying \(\sum_{x\in\mathcal{X}}k_X(x)=0\). Routine calculations show that
\begin{align}
    &\frac{d^2}{d\epsilon^2}\mathcal{L}_{u_X,p_X}(q_X+\epsilon k_X)\Big|_{\epsilon=0}=\label{diff;2}\\
    &\mathbb{E}_{p_X}\!\left[\widetilde{\omega}_{X}^2\,g''\!\left(\frac{q_X}{p_X}\right)-4\widetilde{\omega}_{X}\widetilde{\gamma}\,g'\!\left(\frac{q_X}{p_X}\right)+2\widetilde{\gamma}^{2}g\!\left(\frac{q_X}{p_X}\right)\right],\nonumber
\end{align}
where
\begin{align}
    \widetilde{\omega}_{X}&=\frac{\left(u_X-\mathbb{E}_{q}\left(u_X\right)\right)k_X}{p_X},\\
    \widetilde{\gamma}&=\mathbb{E}_{p_X}\left(\widetilde{\omega}_{X}\right).\label{qw:1}
\end{align}

Now, we wish to design a function \(g(x)\) such that the expression in \eqref{diff;2} is non-positive. Let \(\eta_g\) be a constant. Using \(g'=(g'-\eta_g)+\eta_g\) together with \eqref{qw:1}, we obtain
\begin{align}
    &\frac{d^2}{d\epsilon^2}\mathcal{L}_{u_X,p_X}(q_X+\epsilon k_X)\Big|_{\epsilon=0}=
    \mathbb{E}_{p_X}\!\left[\widetilde{\omega}_{X}^2\,g''\!\left(\frac{q_X}{p_X}\right)\right.\nonumber\\&\left.-4\widetilde{\omega}_{X}\widetilde{\gamma}\left(g'\!\left(\frac{q_X}{p_X}\right)-\eta_g\right)+2\widetilde{\gamma}^{2}g\!\left(\frac{q_X}{p_X}\right)\right]-4\widetilde{\gamma}^2\eta_g.\label{diff;3}
\end{align}
Assuming \(g''(x)\neq 0\), rearranging \eqref{diff;3} yields
\begin{align}
    &\frac{d^2}{d\epsilon^2}\mathcal{L}_{u_X,p_X}(q_X+\epsilon k_X)\Big|_{\epsilon=0}=
    \mathtt{A}_1+2\widetilde{\gamma}^2\mathtt{A}_{2}-4\widetilde{\gamma}^2\eta_g,
\end{align}
where
\begin{align}
\mathtt{A}_1&=\mathbb{E}_{p_X}\!\left[\frac{\left(\widetilde{\omega}_{X}\,g''\!\left(\frac{q_X}{p_X}\right)-2\widetilde{\gamma}\left(g'\!\left(\frac{q_X}{p_X}\right)-\eta_g\right)\right)^2}{g''\!\left(\frac{q_X}{p_X}\right)}\right],\\
 \mathtt{A}_2&=\mathbb{E}_{p_X}\!\left[g\!\left(\frac{q_X}{p_X}\right)+\frac{2\left(g'\!\left(\frac{q_X}{p_X}\right)-\eta_g\right)^2}{-g''\!\left(\frac{q_X}{p_X}\right)}\right].
\end{align}
It is clear that if \(g''(x)<0\) for all \(x>0\), then \(\mathtt{A}_1\leq 0\). Moreover, if we assume that for all \(x>0\) the function \(g(x)\) satisfies
\begin{align}
    g(x)+\frac{2\left(g'(x)-\eta_g\right)^2}{-g''(x)}\leq \theta_1+\theta_2(x-1)\label{iu:1}
\end{align}
for some constants \(\theta_1\) and \(\theta_2\), then \(\mathtt{A}_2\leq \theta_1\). Thus, we have
\begin{align}
    \frac{d^2}{d\epsilon^2}\mathcal{L}(q_X+\epsilon k_X)\Big|_{\epsilon=0}&=
\mathtt{A}_1+2\widetilde{\gamma}^2\mathtt{A}_{2}-4\widetilde{\gamma}^2\eta_g,\nonumber\\
    &\leq 2\widetilde{\gamma}^2(\theta_1-2\eta_g).
\end{align}
By choosing \(\theta_1= 2\eta_g\), we obtain \(\frac{d^2}{d\epsilon^2}\mathcal{L}(q_X+\epsilon k_X)\big|_{\epsilon=0}\leq 0\). 
Recall that \(f\) is a convex function. Since \(x^2f''(x)=g(x)\) and we assume that \(g''(x)<0\), the left-hand side of \eqref{iu:1} is a positive function, and thus \(\mathtt{A}_{2}\) is positive. Since \(\mathtt{A}_{2}\leq \theta_1\), we conclude that \(\eta_g\) must be positive.
 This completes the proof of Lemma~\ref{yhn:1}.
\end{proof}




\subsection{Proof of Corollary~\ref{coro:1}}\label{coro:1:proof}
By assumption, the function \(s(x)\) is positive and concave and \(s''(x)<0\). Moreover, 
\[
\mathtt{K}_{s}(x)\triangleq s(x)+\frac{2\,(s'(x))^2}{-s''(x)}
\]
is a concave function. The first-order inequality for the concave function \(\mathtt{K}_{s}(x)\) yields 
\begin{align*}
   \mathtt{K}_{s}(x)\leq \mathtt{K}_{s}(1)+\mathtt{K}_{s}^{'}(1)(x-1).
\end{align*}
Define \(g(x)=\alpha x+s(x)\) with \(\alpha>0\). For \(\alpha\geq \max\{\mathtt{K}_{s}(1),0\}\), the function \(g(x)\) satisfies the assumptions of Theorem \ref{thm:mutual-f-subadditivity} by choosing \((\eta_g,\ell_g)=\bigl(\alpha,\mathtt{K}_{s}^{'}(1)+\alpha\bigr)\), since
\begin{align*}
    &g(x)+\frac{2\,(g'(x)-\eta_g)^2}{-g''(x)}=\eta_g x+\mathtt{K}_{s}(x)\\&\leq 
    \mathtt{K}_{s}(1)+\eta_g+\ell_g(x-1)\leq 2\eta_g+\ell_{g}(x-1).
\end{align*}
Moreover, \(g''(x)=s''(x)<0\). Also, by choosing \(\alpha=\eta_g>0\), the function \(g(x)\) is positive. This completes the proof.

\subsection{Details of Example~\ref{ex;1}}\label{ex;1:p}
\begin{enumerate}
    \item \textbf{First Case: } The positive concave function \(s(x)\) is given by \(s(x)=b\,x^\gamma\) for \(\gamma\in(0,1)\) and \(b>0\). The corresponding \(\mathtt{K}_{s}(x)\) is
\[
\mathtt{K}_{s}(x)=b\,\frac{1+\gamma}{1-\gamma}\,x^\gamma.
\]
Let \(\alpha\geq \mathtt{K}_{s}(1)=b\,\frac{1+\gamma}{1-\gamma}\). Solving \(x^2f''(x)=\alpha x+s(x)\) with \(f(1)=0\) yields
\[
f(x)=\alpha x\log x+\frac{b}{\gamma(\gamma-1)}\left(x^{\gamma}-1\right)+\widetilde{a}(x-1),
\]
where \(\widetilde{a}\) is an arbitrary parameter. Choosing \((\alpha,\gamma,b,\widetilde{a})=\left(1,\frac{1}{2},\frac{1}{3},\frac{2}{3}\right)\) gives
\[
f(x)=x\log x+\frac{2}{3}\left(\sqrt{x}-1\right)^2.
\]
Since $x\log(x)$ leads to a subadditive information measure, by mixing the above function with $x\log(x)$, we get that
    $f(x)=x\log(x)+\frac{2\alpha}{3}(\sqrt{x}-1)^{2}$ also leads to a subadditive information measure for all $\alpha\in[0,1]$.
\item \textbf{Second Case: } 
Assume that a parametric family \(f_{\gamma}(x)\) leads to a subadditive information measure for any \(\gamma\) in an interval \(\widehat{J}\). Let \(\mu\) be a positive finite measure on \(\widehat{J}\). Then, the mixture
\[
f(x)=\int_{\widehat{J}}f_{\gamma}(x)\,d\mu(\gamma)
\]
 also leads to a subadditive information measure. 
Now, in the previous example, choose \((\alpha,b,\widetilde{a})=\bigl(\gamma(1+\gamma),\,\gamma(1-\gamma),\,0\bigr)\) and let \(\mu\) be the uniform distribution on \((0,1)\). Then, we obtain the function
\begin{align*}
    f(x)&=\int_{0}^{1}\Bigl[\gamma(\gamma+1)\,x\log x+(1-x^{\gamma})\Bigr]\,d\gamma\\&=\frac{5}{6}\,x\log x+1-\frac{x-1}{\log x},
\end{align*}
which is subadditive.
\item \textbf{Third Case: } Consider the positive concave function $s(x)=b\frac{x}{x+1}$ for positive $b$. The corresponding \(\mathtt{K}_{s}(x)\) is
$
\mathtt{K}_{s}(x)=b
$.
Let \(\alpha\geq \mathtt{K}_{s}(1)=b\). Solving \(x^2f''(x)=\alpha x+s(x)\) with \(f(1)=0\) yields
\[
f(x)=(\alpha+b) x\log x-b (x+1)\log\left( \frac{x+1}{2}\right)
+\widetilde{a}(x-1),
\]
where \(\widetilde{a}\) is an arbitrary parameter. Choosing \((\alpha,b,\widetilde{a})=\left(1,1,0\right)\) gives
\[
f(x)=2x\log x- (x+1)\log\left( \frac{x+1}{2}\right).
\]
\item \textbf{Fourth Case: } Consider the positive concave function $s(x)=b\log(1+x)$ for positive $b$. The corresponding \(\mathtt{K}_{s}(x)\) is
$
\mathtt{K}_{s}(x)=b(\log(1+x)+2)
$.
Let \(\alpha\geq \mathtt{K}_{s}(1)=b(2+\log(2))\). Using the formula in equation \eqref{int:1:1}, we get 
\[
f(x)=\alpha x\log(x)+b \int_{1}^{x}(x-t)\frac{\log(1+t)}{t^2} dt.
\]
Choosing \((\alpha,b)=\left(1,\frac{1}{4}\right)\) gives
\[
f(x)= x\log(x)+\frac{1}{4} \int_{1}^{x}(x-t)\frac{\log(1+t)}{t^2} dt.
\]
\end{enumerate}
\subsection{Proof of Theorem~\ref{th:concave:G}}\label{th:concave:G:p}

Since \(G(x)\) is obtained from Lemma~\ref{corol:g}, we deduce that the function \(G^{-1}\!\bigl(G(a)+G(b)\bigr)\) is concave in each of \(a\) and \(b\) separately; see Lemma~\ref{gen-lemma1} and Lemma~\ref{corol:g}. We have the following lemma.

\begin{lemma}\label{lem:pok:1}
  Let \(G\) be an increasing function with \(G(0)=0\). Assume that \(a,b\in[0,m]\) and that \(G^{-1}\!\bigl(G(a)+G(b)\bigr)\) is concave in each of \(a\) and \(b\) separately. Then
   \begin{align}
G^{-1}\!\left(G(a)+G(b)\right)&\geq a+b\nonumber\\&+\frac{G^{-1}\!\left(2G(m)\right)-2m}{m^2}ab.\label{er:44}
   \end{align}
   
\end{lemma}

\begin{proof}
Let \(\lambda=\frac{a}{m}\in[0,1]\) and \(\eta=\frac{b}{m}\in[0,1]\). We write \(a=\lambda m+0\cdot(1-\lambda)\) and \(b=\eta m+0\cdot(1-\eta)\). Concavity of
\(G^{-1}\!\bigl(G(a)+G(b)\bigr)\) with respect to each of its arguments yields
\begin{align*}
    G^{-1}\!\bigl(G(a)+G(b)\bigr)&\geq \lambda\, G^{-1}\!\bigl(G(m)+G(b)\bigr)\\&\quad+(1-\lambda)
    G^{-1}\!\bigl(G(0)+G(b)\bigr)\nonumber\\
    &\geq \lambda\eta\, G^{-1}\!\bigl(G(m)+G(m)\bigr)\nonumber\\
    &\quad+\lambda(1-\eta)\,G^{-1}\!\bigl(G(m)+G(0)\bigr)\nonumber\\
    &\quad+(1-\lambda)\eta\,G^{-1}\!\bigl(G(0)+G(m)\bigr)\nonumber\\
    &\quad+(1-\lambda)(1-\eta)\,G^{-1}\!\bigl(G(0)+G(0)\bigr).
\end{align*}
Using \(G(0)=0\), \(G^{-1}(G(m))=m\), \(\lambda=\frac{a}{m}\), and \(\eta=\frac{b}{m}\) in the above yields
\begin{align*}
    G^{-1}\!\bigl(G(a)+G(b)\bigr)
    &\geq \frac{ab}{m^2}\, G^{-1}\!\bigl(2G(m)\bigr)+m\times\nonumber\\
    &\quad+\left(\frac{a}{m}\Bigl(1-\frac{b}{m}\Bigr)+\frac{b}{m}\Bigl(1-\frac{a}{m}\Bigr)\right)\nonumber\\
    &=a+b+\frac{G^{-1}\!\bigl(2G(m)\bigr)-2m}{m^2}\,ab.
\end{align*}
\end{proof}
\begin{lemma}\label{lem:POI9}
    If \(x^2f''(x)\in \mathcal{T}_{\gamma}\), then
    \begin{align}
   c\leq a+b+\gamma ab,
\end{align}
where
\begin{align*}
  &c\triangleq D_{f}(q_Y q_Z \,\|\, r_Y r_Z),~
  a\triangleq D_{f}(q_Y \,\|\, r_Y),~\\
  &b\triangleq D_{f}(q_Z \,\|\, r_Z), \label{gen-def-1:2:y}
\end{align*}
for any distributions \(q_Y\ll r_Y\) and \(q_Z\ll r_Z\) on finite sets \(\mathcal{Y}\) and \(\mathcal{Z}\), respectively.
\end{lemma}

\begin{proof}
    According to Theorem~3 in \cite{abin2026class:isit}, if the function \(x^2f''\) belongs to the set \(\mathcal{T}^{+}\) defined as
    \begin{align}
    \mathcal{T}^{+}=\Bigl\{g:[0,\infty)\rightarrow [0,\infty)\;\Big|\;&x^2g''(x)-\, g(\alpha)\,g\!\left(\frac{x}{\alpha}\right)\leq 0,\nonumber\\& \forall x>0,~\forall\alpha>0\Bigr\},
\end{align}
then the pair \((\log(1+x),f)\) is divergence-subadditive. Let \(\gamma\) be a positive constant. Since
\[
\mathcal{D}_{\log(1+\gamma x),f}(\cdot\|\cdot)=\mathcal{D}_{\log(1+x),\gamma f}(\cdot\|\cdot),
\]
if \(x^2(\gamma f)''=\gamma\, x^2 f''\in \mathcal{T}^{+}\), then the pair \((\log(1+x),\gamma f)\), and equivalently the pair \((\log(1+\gamma x),f)\), is divergence-subadditive. Let \(x^2f''=g(x)\). The condition \(x^2(\gamma f)''\in \mathcal{T}^{+}\) implies
  \begin{align*}
     &x^2 (\gamma g)''(x)-[\gamma g](\alpha)\,[\gamma g]\!\left(\frac{x}{\alpha}\right)\\
     &=\gamma\left(x^2 g''(x)-\gamma\, g(\alpha)\,g\!\left(\frac{x}{\alpha}\right)\right)\leq 0.
  \end{align*}
Thus, if \(x^2f''\in \mathcal{T}_{\gamma}\), then the pair \((\log(1+\gamma x),f)\) is divergence-subadditive. Divergence-subadditivity of the pair \((\log(1+\gamma x),f)\) implies
  \begin{align}
    \log(1+\gamma c)\leq \log(1+\gamma a)+\log(1+\gamma b),
\end{align}
which is equivalent to
\begin{align}
   c\leq a+b+\gamma ab.
\end{align}
\end{proof}
We are now ready to prove Theorem~\ref{th:concave:G}. Let
\begin{align*}
  &c\triangleq D_{f}(q_Y q_Z \,\|\, r_Y r_Z),~
  a\triangleq D_{f}(q_Y \,\|\, r_Y),~\\
  &b\triangleq D_{f}(q_Z \,\|\, r_Z),
\end{align*}
where the distributions \(q_Y\ll r_Y\) and \(q_Z\ll r_Z\) are defined on finite sets \(\mathcal{Y}\) and \(\mathcal{Z}\), respectively. According to the assumptions of Theorem~\ref{th:concave:G}, since \(x^2f''\in \mathcal{T}_{\gamma}\), Lemma~\ref{lem:POI9} implies
\[c\leq a+b+\gamma ab.\] Since \(G\) is an increasing function, we have
\begin{align*}
   G(c)\leq G(a+b+\gamma ab).
\end{align*}
We want to choose the parameter \(\gamma\) such that 
\[G(a+b+\gamma ab)\leq G(a)+G(b).\]
This is equivalent to
\begin{align}
a+b+\gamma ab\leq G^{-1}\!\bigl(G(a)+G(b)\bigr).\label{poi:1}
\end{align}
Lemma~\ref{lem:pok:1} with \(m=\mathsf{D_m}(f)\) shows that if
\[
0<\gamma\leq \frac{G^{-1}\!\bigl(2G(\mathsf{D_m}(f))\bigr)-2\mathsf{D_m}(f)}{\mathsf{D_m}^{2}(f)},
\]
then inequality \eqref{poi:1} holds. Thus, \(G(c)\leq G(a)+G(b)\). In other words,
\begin{align*}
  &\mathcal{D}_{G,f}(q_Y q_Z \,\|\, r_Y r_Z)\leq \mathcal{D}_{G,f}(q_Y \,\|\, r_Y)+\mathcal{D}_{G,f}(q_Z \,\|\, r_Z).
\end{align*}
Thus, the pair \((G,f)\) is divergence-subadditive.
Note that the concavity assumption on \(G\) only guarantees that \(\gamma\) is positive, since \(G(2x)\leq 2G(x)\) whenever \(G(x)\) is concave on \([0,\infty)\) and \(G(0)=0\). This completes the proof.
\subsection{Details of Example~\ref{ex:con:hh}}\label{ex:con:hh:p}

Let \(\epsilon\in(0,1)\). Consider the non-negative function
\[
\ell_{\epsilon}(x)=\sqrt{x}\left(1+\epsilon \cos(\log x)\right).
\]
Taking the second derivative of \(\ell_{\epsilon}(x)\) yields
\[
\ell_{\epsilon}''(x)=\frac{-1}{4x\sqrt{x}}\left(1+5\epsilon \cos(\log x)\right).
\]
Thus, by choosing \(\epsilon\in\left(\frac{1}{5},1\right)\), the function \(\ell_{\epsilon}(x)\) is neither concave nor convex. We want to show that
\begin{align}
    x^2\ell_{\epsilon}''(x)\leq \gamma\, \ell_{\epsilon}(\alpha)\,\ell_{\epsilon}\!\left(\frac{x}{\alpha}\right). \label{yt:1:1}
\end{align}
Since
\[
\ell_{\epsilon}(\alpha)\,\ell_{\epsilon}\!\left(\frac{x}{\alpha}\right) \geq \sqrt{x}\,(1-\epsilon)^2
\]
and
\[
x^2\ell_{\epsilon}''(x)\leq \frac{\sqrt{x}}{4}(5\epsilon-1),
\]
if we choose a positive \(\gamma\geq \frac{5\epsilon-1}{4(1-\epsilon)^2}\), then inequality \eqref{yt:1:1} holds. Thus, \(\ell_{\epsilon}(x)\in \mathcal{T}_{\gamma}\). Solving the equation \(x^2f_{\epsilon}''(x)=\ell_{\epsilon}(x)\) yields
\[
f_{\epsilon}(x)=-4\sqrt{x}-\frac{4\epsilon}{5}\sqrt{x}\cos(\log x)+\widetilde{a}x+\widetilde{b},
\]
where the parameters \(\widetilde{a}\) and \(\widetilde{b}\) are chosen so that \(f_{\epsilon}(1)=0\). We choose \(\widetilde{a}=\widetilde{b}=2+\frac{2\epsilon}{5}\)  and obtain
\[
f_{\epsilon}(x)=(2+0.4\epsilon)\left(\sqrt{x}-1\right)^2+0.8\epsilon \sqrt{x}\left(1-\cos(\log x)\right).
\]
Using Lemma~\ref{domain:lem}, we have \(\mathsf{D_m}(f_{\epsilon})=4+\frac{4\epsilon}{5}\). To guarantee that the pair \((x^p,f_{\epsilon}(x))\) is divergence-subadditive, we choose the parameter \(\gamma\) as follows:
\begin{align}
    \gamma&= \frac{G^{-1}\!\bigl(2G(\mathsf{D_m}(f))\bigr)-2\mathsf{D_m}(f)}{\mathsf{D_m}^{2}(f)}\nonumber\\
    &=\frac{5(2^{1/p}-2)}{20+4\epsilon}.
\end{align}
Set \(\widetilde{a}_p=5\bigl(2^{1/p}-2\bigr)\). On the other hand, we required \(\gamma\geq \frac{5\epsilon-1}{4(1-\epsilon)^2}\). Thus,  our goal is to choose \(\epsilon\in\left(\frac{1}{5},1\right)\) such that
\begin{align}
   \frac{\widetilde{a}_p}{20+4\epsilon}\geq \frac{5\epsilon-1}{4(1-\epsilon)^2}. \label{mji:1}
\end{align}
Let \(\Theta_p\) be the positive root of the equality case of inequality \eqref{mji:1}. We have
\[
\Theta_p=\frac{-12-\widetilde{a}_p+\sqrt{24\widetilde{a}_p+169}}{5-\widetilde{a}_p}.
\]
It is easy to see that if \(\epsilon\in\left(\frac{1}{5},\Theta_p\right]\), inequality \eqref{mji:1} holds. Thus, the pair \((x^p,f_{\epsilon}(x))\) is divergence-subadditive.

\subsection{Proof of Theorem~\ref{th:convex:G}}\label{th:convex:G:p}
We begin the proof with the following lemma.
\begin{lemma}\label{lem:POI99}
    If \(x^2f''(x)\in \mathcal{T}_{\gamma}^{-}\), then
    \begin{align}
   c\leq a+b-\gamma ab,
\end{align}
where
\begin{align*}
  &c\triangleq D_{f}(q_Y q_Z \,\|\, r_Y r_Z),~
  a\triangleq D_{f}(q_Y \,\|\, r_Y),~\\
  &b\triangleq D_{f}(q_Z \,\|\, r_Z), \label{gen-def-1:2:z}
\end{align*}
for any distributions \(q_Y\ll r_Y\) and \(q_Z\ll r_Z\) on finite sets \(\mathcal{Y}\) and \(\mathcal{Z}\), respectively.
\end{lemma}

\begin{proof}
    According to Theorem~4 in \cite{abin2026class:isit}, if the function \(x^2f''\) belongs to the set \(\mathcal{T}^{-}\) defined as
    \begin{align}
    \mathcal{T}^{-}=\Bigl\{g:[0,\infty)\rightarrow [0,\infty)\;\Big|\;&x^2g''(x)+\, g(\alpha)\,g\!\left(\frac{x}{\alpha}\right)\leq 0,\nonumber\\& \forall x>0,~\forall\alpha>0\Bigr\},
\end{align}
then the pair \((-\log(1-x),f)\) is divergence-subadditive. Let \(\gamma\) be a positive constant. Since
\[
\mathcal{D}_{-\log(1-\gamma x),f}(\cdot\|\cdot)=\mathcal{D}_{-\log(1-x),\gamma f}(\cdot\|\cdot),
\]
if \(x^2(\gamma f)''=\gamma\, x^2 f''\in \mathcal{T}^{-}\), then the pair \((-\log(1-x),\gamma f)\), and equivalently the pair \((-\log(1-\gamma x),f)\), is divergence-subadditive. Let \(x^2f''=g(x)\). The condition \(x^2(\gamma f)''\in \mathcal{T}^{-}\) implies
  \begin{align*}
     &x^2 (\gamma g)''(x)+[\gamma g](\alpha)\,[\gamma g]\!\left(\frac{x}{\alpha}\right)\\
     &=\gamma\left(x^2 g''(x)+\gamma\, g(\alpha)\,g\!\left(\frac{x}{\alpha}\right)\right)\leq 0.
  \end{align*}
Thus, if \(x^2f''\in \mathcal{T}_{\gamma}^{-}\), then the pair \((-\log(1-\gamma x),f)\) is divergence-subadditive. Divergence-subadditivity of the pair \((-\log(1-\gamma x),f)\) implies
  \begin{align}
    -\log(1-\gamma c)\leq -\log(1-\gamma a)-\log(1-\gamma b),
\end{align}
which is equivalent to
\begin{align}
   c\leq a+b-\gamma ab.
\end{align}
\end{proof}
We are now ready to prove Theorem~\ref{th:convex:G}. Let
\begin{align*}
  &c\triangleq D_{f}(q_Y q_Z \,\|\, r_Y r_Z),~
  a\triangleq D_{f}(q_Y \,\|\, r_Y),~\\
  &b\triangleq D_{f}(q_Z \,\|\, r_Z),
\end{align*}
where the distributions \(q_Y\ll r_Y\) and \(q_Z\ll r_Z\) are defined on finite sets \(\mathcal{Y}\) and \(\mathcal{Z}\), respectively. According to the assumptions of Theorem~\ref{th:convex:G}, since \(x^2f''\in \mathcal{T}_{\gamma}^{-}\), Lemma~\ref{lem:POI99} implies
\[
c\leq a+b-\gamma ab.
\]
Since \(G\) is an increasing function, we have
\begin{align*}
   G(c)\leq G(a+b-\gamma ab).
\end{align*}
We want to choose the parameter \(\gamma\) such that
\[
G(a+b-\gamma ab)\leq G(a)+G(b).
\]
This is equivalent to
\begin{align}
a+b-\gamma ab\leq G^{-1}\!\bigl(G(a)+G(b)\bigr).\label{poi:18}
\end{align}
Lemma~\ref{lem:pok:1} with \(m=\mathsf{D_m}(f)\) shows that if
\[
\gamma\geq \frac{2\mathsf{D_m}(f)-G^{-1}\!\bigl(2G(\mathsf{D_m}(f))\bigr)}{\mathsf{D_m}^{2}(f)},
\]
then inequality \eqref{poi:18} holds. Thus, \(G(c)\leq G(a)+G(b)\). In other words,
\begin{align*}
  &\mathcal{D}_{G,f}(q_Y q_Z \,\|\, r_Y r_Z)\leq \mathcal{D}_{G,f}(q_Y \,\|\, r_Y)+\mathcal{D}_{G,f}(q_Z \,\|\, r_Z).
\end{align*}
Thus, the pair \((G,f)\) is divergence-subadditive.
Note that the convexity assumption on \(G\) only guarantees that \(\gamma\) is greater than a positive constant, since \(G(2x)\geq 2G(x)\) whenever \(G(x)\) is convex on \([0,\infty)\) and \(G(0)=0\). This completes the proof.
\subsection{Details of Example~\ref{ex:convex:rt}}\label{ex:convex:rt:p}

Let \(\epsilon\in\left(0,\frac{1}{25}\right)\). Let \(m_\epsilon\) be a positive constant to be determined later. Consider the non-negative function
\[
\ell_{\epsilon}(x)=\frac{m_{\epsilon}}{4(1+\epsilon)}\sqrt{x}\left(1+5\epsilon \cos(\log x)\right).
\]
Taking the second derivative of \(\ell_{\epsilon}(x)\) yields
\[
\ell_{\epsilon}''(x)=\frac{-m_{\epsilon}}{16(1+\epsilon)\,x\sqrt{x}}\left(1+25\epsilon \cos(\log x)\right).
\]
The function \(\ell_{\epsilon}(x)\) is concave. We want to show that
\begin{align}
    -x^2\ell_{\epsilon}''(x)\geq \gamma\, \ell_{\epsilon}(\alpha)\,\ell_{\epsilon}\!\left(\frac{x}{\alpha}\right). \label{yt:1:1:u}
\end{align}
Since
\[
\ell_{\epsilon}(\alpha)\,\ell_{\epsilon}\!\left(\frac{x}{\alpha}\right) \leq \frac{m_{\epsilon}^2}{16(1+\epsilon)^2}\sqrt{x}\,(1+5\epsilon)^2
\]
and
\[
x^2\ell_{\epsilon}''(x)\geq \frac{m_{\epsilon}\sqrt{x}}{16(1+\epsilon)}(1-25\epsilon),
\]
if we choose a positive \(\gamma\) such that \(\gamma m_{\epsilon}\leq \frac{(1+\epsilon)(1-25\epsilon)}{(1+5\epsilon)^2}\), then inequality \eqref{yt:1:1:u} holds. Thus,  \(\ell_{\epsilon}(x)\in \mathcal{T}_{\gamma}^{-}\). Solving the equation \(x^2f_{\epsilon}''(x)=\ell_{\epsilon}(x)\) yields
\[
f_{\epsilon}(x)=\frac{-m_{\epsilon}\sqrt{x}}{1+\epsilon}\Bigl[1+\epsilon\cos(\log x)\Bigr]+\widetilde{a}x+\widetilde{b},
\]
where the parameters \(\widetilde{a}\) and \(\widetilde{b}\) are chosen so that \(f_{\epsilon}(1)=0\). We choose \((\widetilde{a},\widetilde{b})=(0,m_{\epsilon})\) and obtain
\[
f_{\epsilon}(x)=m_{\epsilon}\left(1-\frac{\sqrt{x}\,\bigl(1+\epsilon\cos(\log x)\bigr)}{1+\epsilon}\right).
\]
Using Lemma~\ref{domain:lem}, we have \(\mathsf{D_m}(f_{\epsilon})=m_{\epsilon}\).
Note that for \(G(x)=\frac{1}{2}\log\left(\frac{1+x}{1-x}\right)\), the quantity \(\dfrac{2x-G^{-1}(2G(x))}{x^2}\) is equal to \(\frac{2x}{x^2+1}\).
To guarantee that the pair \(\bigl(\frac{1}{2}\log\bigl(\frac{1+x}{1-x}\bigr),f_{\epsilon}(x)\bigr)\) is divergence-subadditive, we choose the parameter \(\gamma\) as follows:
\begin{align}
    \gamma&= \frac{2\mathsf{D_m}(f)-G^{-1}\!\bigl(2G(\mathsf{D_m}(f))\bigr)}{\mathsf{D_m}^{2}(f)}\nonumber\\
    &=\frac{2m_{\epsilon}}{1+m_{\epsilon}^{2}}.
\end{align}
On the other hand, we required \(\gamma m_{\epsilon}\leq \frac{(1+\epsilon)(1-25\epsilon)}{(1+5\epsilon)^2}\). Thus, our goal is to choose \(\epsilon\in\left(0,\frac{1}{25}\right)\) such that
\begin{align}
  \frac{2m_{\epsilon}^{2}}{1+m_{\epsilon}^{2}}\leq \frac{(1+\epsilon)(1-25\epsilon)}{(1+5\epsilon)^2}. \label{mji:22}
\end{align}
We choose \(m_\epsilon\) so that equality holds in \eqref{mji:22}. We have
\[
m_\epsilon=\sqrt{\frac{(1+\epsilon)(1-25\epsilon)}{75\epsilon^2+44\epsilon+1}},
\qquad
\gamma=\frac{2m_{\epsilon}}{m_{\epsilon}^2+1}.
\]
Thus, all the required inequalities hold. Hence, the pair \(\bigl(\frac{1}{2}\log\bigl(\frac{1+x}{1-x}\bigr),f_{\epsilon}(x)\bigr)\) is divergence-subadditive.
\subsection{Details of Example~\ref{neither:conv:1}}\label{neither:conv:1:p}

As shown in Appendix~\ref{ex:convex:rt:p}, the function \(x^2f_{\epsilon}''(x)\) defined in Example~\ref{ex:convex:rt} belongs to \(\mathcal{T}_{\gamma}^{-}\) provided that we can find a positive constant \(\gamma\) such that
\[
\gamma m_{\epsilon}\leq \frac{(1+\epsilon)(1-25\epsilon)}{(1+5\epsilon)^2}.
\]
Moreover, according to the proof of Theorem~\ref{th:convex:G}, we do not require \(G\) to be convex; if we find a positive \(\gamma\) such that
\[
\gamma\geq \frac{2\mathsf{D_m}(f_{\epsilon})-G^{-1}\!\bigl(2G(\mathsf{D_m}(f_{\epsilon}))\bigr)}{\mathsf{D_m}^{2}(f_{\epsilon})},
\]
then the pair \((G,f_{\epsilon})\) is divergence-subadditive, since \(x^2f_{\epsilon}''(x)\in \mathcal{T}_{\gamma}^{-}\). Thus, if we can find a positive \(\gamma\) such that the following two inequalities hold, then subadditivity follows:
\begin{align}
  \gamma m_{\epsilon}&\leq \frac{(1+\epsilon)(1-25\epsilon)}{(1+5\epsilon)^2},\label{uytg:8}\\
  \gamma&\geq \frac{2\mathsf{D_m}(f_{\epsilon})-G^{-1}\!\bigl(2G(\mathsf{D_m}(f_{\epsilon}))\bigr)}{\mathsf{D_m}^{2}(f_{\epsilon})}.\label{bhu:1}
\end{align}
For the function \(f_\epsilon\) in this example, \(\mathsf{D_m}(f_\epsilon)=m_{\epsilon}<1\). For the function \(G(x)=\frac{1}{4}\log\left(\frac{3(1+x)}{3-x}\right)\), the quantity
\[
\frac{2x-G^{-1}\!\bigl(2G(x)\bigr)}{x^2}
\]
is equal to \(\frac{2(x-1)}{x^2+3}\). Since \(m_{\epsilon}<1\), inequality \eqref{bhu:1} holds automatically for every positive parameter \(\gamma\). Thus, by choosing
\[
\gamma = \frac{(1+\epsilon)(1-25\epsilon)}{m_{\epsilon}(1+5\epsilon)^2},
\]
inequality \eqref{uytg:8} holds. Therefore, the pair
\[
\left(\frac{1}{4}\log\left(\frac{3(1+x)}{3-x}\right),f_{\epsilon}\right)
\]
is divergence-subadditive.

\end{document}